%% file: Main_file.tex
\documentclass[11pt, draftcls, onecolumn]{IEEEtran}

 \usepackage{setspace}
\input{Preamble}

\graphicspath{{./Figures/}}

\title{Graph Compression with Side Information at the Decoder}

\author{Praneeth Kumar Vippathalla, Mihai-Alin Badiu and Justin P. Coon \thanks{ The authors are with the Department of Engineering Science, University of Oxford, OX1 3PJ Oxford, U.K. (e-mail: praneeth.vippathalla@eng.ox.ac.uk; mihai.badiu@eng.ox.ac.uk;
justin.coon@eng.ox.ac.uk)}
\thanks{This research was funded in whole or in part by the U. S. Army Research Laboratory and the U. S. Army Research Office (W911NF-22-1-0070), and EPSRC (EP/T02612X/1). For the purpose of Open Access, the authors have applied a CC BY public copyright licence to any Author Accepted Manuscript (AAM) version arising from this submission.}}

\begin{document}

\maketitle

\begin{abstract}
  \input{Abstract}
\end{abstract} 


\section{Introduction} \label{sec:introduction}
\input{Introduction}

\section{Graph Compression with Side Information} \label{sec:graph_compression_problem}
\input{Problem_graph_compression}

\section{Minimum Compression Rate}
\label{sec:min_comp_rate}
\input{Min_comp_rate}

\section{Unweighted Graphs}
\label{sec:unweighted_graphs}
\input{Unweighted_graphs}

\section{Discussion}
\label{sec:discussion}
\input{Discussion}

\bibliographystyle{IEEEtran}
\bibliography{IEEEabrv,References}

\end{document}

%% file: Preamble.tex
\usepackage{tikz}
\usetikzlibrary{shapes, arrows, positioning, calc}
\usepackage{graphicx,xcolor}
\usepackage[noadjust]{cite}
\usepackage{enumitem}
\usepackage{setspace}

\usepackage{amsmath,amssymb,amsthm,amsfonts}

\usepackage{bbm}
\usepackage{bm}
\usepackage{mathrsfs}
\usepackage{dsfont}

\newtheorem{theorem}{Theorem}

\newtheorem{lemma}{Lemma}

\theoremstyle{definition}

\usepackage{caption}
\usepackage{subcaption}
\allowdisplaybreaks

\newcommand{\ea}{A}
\newcommand{\eb}{B}
\newcommand{\eea}{a}
\newcommand{\eeb}{b}

\newcommand{\ga}{G_{\scriptscriptstyle a}}

\newcommand{\gga}{g_{\scriptscriptstyle a}}

\newcommand{\gb}{G_{\scriptscriptstyle b}}
\newcommand{\ggb}{g_{\scriptscriptstyle b}}

\newcommand{\sta}{S_{\scriptscriptstyle a}}
\newcommand{\starel}{s_{\scriptscriptstyle a}}

\newcommand{\stb}{S_{\scriptscriptstyle b}}
\newcommand{\stbrel}{s_{\scriptscriptstyle b}}

\newcommand{\pgagb}{P_{\scriptscriptstyle G_a, G_b}}
\newcommand{\pgasb}{P_{\scriptscriptstyle G_a, S_b}}

\newcommand{\pgacgb}{P_{\scriptscriptstyle G_a|G_b}}
\newcommand{\pgacsb}{P_{\scriptscriptstyle G_a|S_b}}
\newcommand{\psacgb}{P_{\scriptscriptstyle S_a|G_b}}
\newcommand{\psacsb}{P_{\scriptscriptstyle S_a|S_b}}

\newcommand{\pgb}{P_{\scriptscriptstyle G_b}}

\newcommand{\psb}{P_{\scriptscriptstyle S_b}}

\newcommand{\ua}{U_{\scriptscriptstyle a}}
\newcommand{\ub}{U_{\scriptscriptstyle b}}
\newcommand{\hua}{\hat{U}_{\scriptscriptstyle a}}

\newcommand{\uua}{u_{\scriptscriptstyle a}}
\newcommand{\uub}{u_{\scriptscriptstyle b}}
\newcommand{\huua}{\hat{u}_{\scriptscriptstyle a}}
\newcommand{\puacub}{P_{\scriptscriptstyle U_a|U_b}}
\newcommand{\puaub}{P_{\scriptscriptstyle U_a, U_b}}

\newcommand{\pab}{p_{\scriptscriptstyle ab}}

\newcommand{\pzerozero}{p_{\scriptscriptstyle 00}}
\newcommand{\pzeroone}{p_{\scriptscriptstyle 01}}
\newcommand{\ponezero}{p_{\scriptscriptstyle 10}}
\newcommand{\poneone}{p_{\scriptscriptstyle 11}}

%% file: Abstract.tex
In this paper, we study the problem of graph compression with side information at the decoder. The focus is on the situation when an unlabelled graph (which is also referred to as a structure) is to be compressed or is available as side information. For correlated Erd\H{o}s-R\'enyi weighted random graphs, we give a precise characterization of the smallest rate at which a labelled graph or its structure can be compressed with aid of a correlated labelled graph or its structure at the decoder. We approach this problem by using the entropy-spectrum framework and establish some convergence results for conditional distributions involving structures, which play a key role in the construction of an optimal encoding and decoding scheme.  Our proof essentially uses the fact that, in the considered correlated Erd\H{o}s-R\'enyi model, the structure retains most of the information about the labelled graph. Furthermore, we consider the case of unweighted graphs and present how the optimal decoding can be done using the notion of graph alignment.

%% file: Introduction.tex
Today, graph-structured data can be found in numerous contexts. Advancements in data collection methods and applications have led to a huge growth in the amount of graphical data that is generated. Some examples of graphical data include social interaction networks, biological networks, protein-protein interaction networks, and web graphs. Graphical data represents how entities are related to each other in a pairwise manner. In practice, we often encounter unlabelled graphs along with labelled graphs\footnote{We use the terms ``structure'' and ``unlabelled graph'' interchangeably throughout the paper.}.  Where labelled graphs contain, for example, names of individuals along with the data pertaining to them, unlabelled graphs mainly appear in the context of privacy, where  graphical data is revealed to the public for scientific purposes with the identities of the nodes hidden. Such anonymized graphs can then be analyzed without knowing the true identities of the nodes, thus ensuring the privacy of the nodes. 

Although graphical data is useful for representing interrelations between entities, the size of such data is often enormous, which can make storage and processing difficult. For example, at the time of writing, the number of Facebook users is close to 3 billion.  The prevalence of such large graphical data sets spurred interest in the information-theoretic study of graph compression in recent years. Many works, such as \cite{choi_structural_entropy, Abbe2016GraphClusters, delgosha_universal_compression, alankrita_universal_sbm, mihai2021structural, bustin_lossy, martin_rate_distortion_sbm,lossy_spatial_24}, studied the compression of a labelled graph or, alternatively, an unlabelled graph under various assumptions on the underlying random  graph model and constraints on compression fidelity. Another scenario that is of importance is the compression of graphs with side information. Here one would like to compress a labelled or an unlabelled graph using a correlated labelled or  unlabelled graph. For instance, one could possibly compress a labelled graph at a smaller rate with aid of a publicly available anonymized graph that is correlated with the labelled graph. This motivates us to study the graph compression problem with side information from an information theoretic perspective. The broad question of interest is the following: \emph{what is the minimum achievable compression rate?}

Before undertaking to answer this question, a few issues must be addressed. The first relates to the use of the structure as side information. A natural question would be whether there is any tangible gain in the compression rate when compared to not using any side information. After all, the structure ignores the node label information. We will see later in the paper that most of the information about the random graph that we consider is contained in its structure and it can be fruitfully exploited as side information.

Another issue relates to the location of the side information: encoder and decoder, or only at the decoder. In the situations of compression involving structures, the availability of the side information at only the decoder is the most interesting. For example, if a structure and a correlated graph
are present at the encoder, then the correlated graph can be used to recover the labels of the structure, which is usually referred to as \emph{deanyomization}. There are already well-known techniques available for deanonymization \cite{cullina_improved, pedarsani_2011_privacy, narayan_deanonymize, shirani_correlated_graph_matching}. The resulting deanonymized labelled graph can now be used to compress the correlated labelled graph. However, if the structure is only available at the decoder and there are no means of deanonymizing it, then it must be used without the knowledge of the labels. 

The problem of graph compression using a correlated unlabelled graph at the decoder was introduced and studied recently in \cite{nikpey_graph_side}. The authors of that work considered correlated unweighted graphs that are generated independently by sampling edges from an underlying Erd\H{o}s-R\'enyi graph with a fixed sampling probability. Under this model, an achievable compression rate was derived.  The exact characterization of the optimal compression rate was left as an open problem.

\subsection{Summary of Contributions}
We consider the general setting of weighted graphs, where the weight (or mark) on an edge represents the degree of association of the nodes connected by that edge. We also consider the correlated Erd\H{o}s-R\'enyi model for generating two labelled graphs, from which the respective unlabelled graphs are generated. 

Our main result is the precise characterization of the minimum achievable rate of compression of a labelled or an unlabelled (weighted) graph with a labelled or an unlabelled (weighted) graph as side information at the decoder. This general result resolves the open problem of \cite{nikpey_graph_side}. The main obstacle in the derivation is posed by the structure as it does not lend itself to a simpler analysis. The main feature of independence across edges in an  Erd\H{o}s-R\'enyi labelled random graph is no longer useful in dealing with the structure as multiple labelled graphs correspond to a single structure.

In order to deal with structures, we take  the entropy-spectrum approach of \cite{han_book}. Roughly speaking, in this approach, the convergence properties of the normalized negative logarithm  of the probability distribution of a general source are used to identify the compression rate.  In the paper, we prove a convergence result related to the conditional probability distribution involving structures. To prove this result, we first argue that the maximum of $\pgagb(\ga, \pi(\gb))$ over all permutations $\pi$ of vertex labels  of $\gb$ is close in some sense to $\pgagb(\ga, \gb)$  for correlated Erd\H{o}s-R\'enyi random graphs $\ga$ and $\gb$. Then, we combine this with the fact that not many labelings are possible for a structure compared to the number of possible labelled graphs; there are at most $n!$ distinct labelings of a graph, where $n$ is the number of vertices, while there are $2^{\binom{n}{2}}$ labelled graphs.

The convergence result on the conditional probability distribution allows us to construct an encoding and decoding scheme for graph compression with side information at the decoder. It turns out that the rate achieved by this scheme is optimal. Finally, in the special case of unweighted graphs, we give a way to describe the optimal decoder using the notion of graph alignment.

\subsection{Notation}
 For $k \in \mathbb{N}$, we use the notation $[k]$ to denote the set $\{1,\ldots,k\}$. Throughout the paper, $\log$ stands for the base-$2$ logarithm. The set of all permutations defined over $[n]$ is denoted by $\mathcal{S}_n$. In this paper, we use uppercase letters for random variables (e.g., $A, B, G, S,$ and $U$) and lowercase letters for their realizations (e.g., $a, b, g, s,$ and $u$). Given a random variable $X$ taking values in a finite set $\mathcal{X}$ with probability distribution $P_{X}$, the random variable $P_{X}(X)$ represents the evaluation of the function $P_{X}(\cdot)$ at a random point $X$. The cardinality of a finite set $\mathcal{X}$ is written as $|\mathcal{X}|$. A sequence of random variables $\{X_n\}_{n\geq 1}$ is said to \emph{converge to $X$ in probability} as $n \to \infty$ if $\mathbb{P}\left(|X_n-X|>\delta\right) \rightarrow 0$ as $n \to \infty$ for any $\delta>0$. 
 
\subsection{Organization}
The rest of the paper is organized as follows. In Section~\ref{sec:graph_compression_problem}, we give the necessary preliminaries and then formally introduce the problem of graph compression with side information for correlated Erd\H{o}s-R\'enyi random graphs. Section~\ref{sec:min_comp_rate} contains the main results related to the minimum achievable compression rate for the general case of weighted graphs. In Section~\ref{sec:unweighted_graphs}, we address the unweighted graph case and present the decoding scheme in terms of the graph alignment concept. We finally give a discussion of potential extensions in Section~\ref{sec:discussion}.

%% file: Problem_graph_compression.tex
\subsection{Preliminaries}
Let $G=(V, E, w)$ be a (weighted) graph with the vertex set $V=[n]$ containing $n$ vertices, the edge set $E$, and the edge-weight function $w:E \to [k]$. The weight $w(e)$ of an edge $e$ connecting two vertices $i, j \in V$ may indicate the strength of the connection between $i$ and $j$. The weights of the edges are sometimes referred to as marks of the edges. We restrict our attention to simple and undirected weighted graphs. It is often useful to think of a graph in terms of its adjacency matrix representation. The entries of the adjacency matrix $\mathbf{A}=(a_{i,j})_{i,j \in [n]}$ of a graph are defined as $a_{i,j}= 0$ if there is no edge between vertices $i$ and $j$, and $a_{i,j}= w(e)$ if $e \in E$ is the edge between vertices $i$ and $j$. As the graphs are simple and undirected, the adjacency matrix is symmetric and the entries $(a_{i,j}:1\leq i < j \leq n)$ completely specify the graph $G$. Using this equivalence, it is customary to represent a graph by a string of $\binom{n}{2}$ symbols over the alphabet $\{0,1, \ldots, k\}$. In the case of random graphs, a probability distribution of graphs $P_{G}$ uniquely specifies a joint distribution for $(a_{i,j}:1\leq i < j \leq n)$, and vice versa. 

A graph operation that is mainly used later in this paper is the permutation of the vertex labels. Let $\pi:[n] \to [n]$ be a permutation map of the vertex labels. We use the notation $\pi(G)$ to denote the graph obtained after applying $\pi$ to the vertex labels of $G$. If $a^{\pi}_{i,j}$ and $a_{i,j}$ are the $(i,j)$th entries of the adjacency matrices of $\pi(G)$ and $G$, respectively, then $a^{\pi}_{i,j}= a_{\pi^{-1}(i), \pi^{-1}(j)}$.

A structure (or an unlabelled graph) $S$ of a weighted graph $G$ is obtained by removing the vertex labels of the graph. Precisely speaking, a structure $S$ of a graph $G$ is the set of (weighted) graphs isomorphic to $G$. This set contains all the distinct graphs that are obtained by permuting the vertex labels of $G$. Thus, the probability distribution $P_S$ of structures induced by a probability distribution $P_G$ of graphs is given by
\begin{align}
    P_S(s) =\sum_{g \cong s} P_G(g),
\end{align}
where we use the notation $G \cong S$ to mean that the graph $G$ has the structure $S$. It is well-known \cite{choi_structural_entropy} that the number of distinct graphs having a structure $S$ is $\frac{n!}{|\operatorname{Aut}(S)|}$, where $\operatorname{Aut}(S)$ denotes the automorphism group $\operatorname{Aut}(G)$ of a graph $G$ with the structure $S$. The automorphism group of $G$ contains all the permutations of vertex labels which when applied leaves $G$ unchanged. 

In the context of graph anonymization, the structure of a graph is revealed after removing the vertex labels of the graph. See Fig.~\ref{fig:anony:example} for a toy example. Another equivalent way of anonymizing a graph is by applying a uniform random  permutation $\Pi$ to a graph $G$ and revealing $\Pi(G)$ publicly. These two methods are equivalent because $S$ can be obtained by removing the vertex labels of $\Pi(G)$, and $\Pi(G)$ can be obtained by randomly assigning a labelling to $S$. From the point of view of privacy, both these representations carry equal amounts of information about the original vertex labels. So, we will work with structures in this paper. 
\begin{figure}[t]
     \centering
     \begin{subfigure}[b]{0.6\columnwidth}
         \centering
         \resizebox{0.5\columnwidth}{!}{\input{Figures/anonymization_example}}
         \caption{Vertex label anonymization}
         \label{fig:anony}
     \end{subfigure}
     \bigskip
    
     \begin{subfigure}[b]{0.6\columnwidth}
         \centering
         \resizebox{\columnwidth}{!}{\input{Figures/deanoymization_example}}
         \caption{Given $s$, each of the three labelings is equally likely.}
         \label{fig:deanony}
     \end{subfigure}
        \caption{A toy example}
        \label{fig:anony:example}
\end{figure}
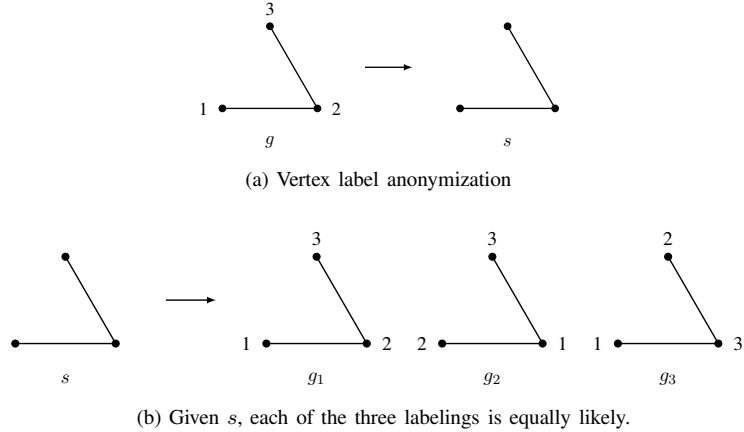

\subsection{Correlated Erd\H{o}s-R\'enyi (CER) Random Graphs}
A fundamental random graph model that is commonly used \cite{cullina_improved, shirani_correlated_graph_matching} in the study of correlated graphs is the \emph{Correlated Erd\H{o}s-R\'enyi (CER)} model. Let $\ga$ and $\gb$ be a pair of correlated weighted graphs with the same vertex set $[n]$. Assume that $[k_a]$ and $[k_b]$ are the sets from which the weight functions $w_a$ and $w_b$ of graphs $G_a$ and $G_b$ take values, respectively. Let $\mathbf{A}=(a_{i,j})$ and $\mathbf{B}=(b_{i,j})$ be the adjacency matrices of $\ga$ and $\gb$, respectively. 
Given a joint distribution
$P_{\ea, \eb}$ on $\{0,1,\ldots, k_a\} \times \{0,1, \ldots, k_b\}$,  the joint distribution of graphs $(\ga, \gb)$ under the CER model is given by 
\begin{align}
    P_{\ga, \gb}(\gga, \ggb) \triangleq \prod_{i<j} P_{\ea, \eb}(\eea_{i,j}, \eeb_{i,j}).
\end{align} 
 In this model, the way in which two vertices $i$ and $j$ are connected in $\ga$ is correlated with how they are connected in $\gb$, and the connections are independent across the vertex pairs in both $\ga$ and $\gb$.

We denote by $\sta$ and $\stb$ the structures of the graphs $\ga$ and $\gb$, respectively. As structures are functions of labelled graphs, the CER model induces the joint distribution $P_{\ga, \gb, \sta, \stb}(\gga, \ggb, \starel, \stbrel) =  P_{\ga, \gb}(\gga, \ggb) \mathds{1}(\gga \cong \starel)\mathds{1}(\ggb \cong \stbrel).
$

\subsection{Problem Formulation}

In this section, we formulate the problem of graph compression with side information at the decoder for the CER model. We will consider the compression of a labelled graph $\ga$ or an unlabelled graph $\sta$  with the correlated side information $\gb$ or  $\stb$ at the decoder.  We use the notation $\ua$ to denote the random object that will be compressed, and the notation $\ub$ to denote the correlated side information at the decoder. In our case, $\ua$ will be either $\ga$ or $\sta$, and $\ub$ will be either $\gb$ or $\stb$.

Let us formally define the graph compression problem with side information at the decoder. Suppose that we are given $(\ua, \ub)$ with joint distribution $P_{\ua, \ub}$. Here, the number of vertices in $\ua$ and $\ub$ is $n$. Though we omitted it in the notation, the random objects implicitly depend on $n$. Our goal is to determine the minimum compression rate possible for encoding $\ua$ with side information $\ub$ at the decoder.
An \emph{encoder} $\phi_n$ operating at a rate $R \geq 0$ is a deterministic function whose input is $\ua$ and whose output is a message $M_n \in \big[2^{{\binom{n}{2}R}}\big]$, i.e., $M_n=\phi_n(\ua)$. A \emph{decoder} $\psi_n$ is a deterministic function that takes a message  $M_n  \in \big[2^{{\binom{n}{2}R}}\big]$ and $\ub$ as inputs, and produces a graph $\hua$ as the output, i.e., $\hua = \psi_n(M_n, \ub)$. We say that  $R$ is an \emph{achievable compression rate} if there exists a sequence of encoder and decoder pairs $\{(\phi_n, \psi_n )\}_{n \geq 1}$ such that 
\begin{align}
    \mathbb{P}\left(\hua \neq \ua\right) = \mathbb{P}\left(\psi_n(\phi_n(\ua),\ub) \neq \ua \right) \rightarrow 0
\end{align}
as $n \to \infty$.
The \emph{minimum achievable compression rate} $R^*$ is defined as
\begin{align}
    R^*\triangleq \inf\{R: R \text{ is an achievable compression rate} \}.
\end{align}
This model is illustrated in Fig.~\ref{fig:sys_model}.

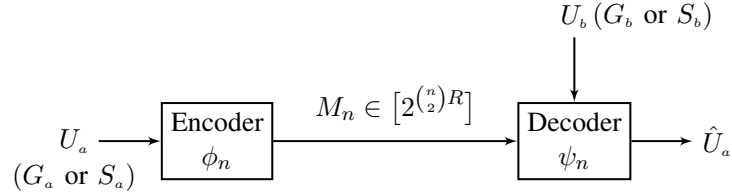
\begin{figure}[t]
\centering
\resizebox{0.95\width}{!}{\input{Figures/sys_model}}
\caption{Graph compression with side information at the decoder. The graphs $(\ga,\gb)$ are correlated Erd\H{o}s-R\'enyi graphs, and $(\sta, \stb)$ are the respective structures. $\ua$ can be $\ga$ or $\sta$, and $\ub$ can be $\gb$ or $\stb$.  The encoder compresses $\ua$ into a message $M_n$ of rate $R$. Using the correlated side information $\ub$, the decoder tries to recover $\ua$.}
\label{fig:sys_model}
 \end{figure}

%% file: Figures/anonymization_example.tex
\begin{tikzpicture}[node distance=4cm,auto,>=latex]
\draw[fill=black] (0,0) circle (2pt);
\draw[fill=black] (2,0) circle (2pt);
\draw[fill=black] (1,1.73) circle (2pt);

\node at (-0.4,0) {1};
\node at (2.4,0) {2};
\node at (1,2.1) {3};
\node at (1,-0.7) {$g$};

\draw[thick] (0,0) -- (2,0) -- (1,1.73);

\draw[->, thick] (3,0.866) to (4,0.866);

\draw[fill=black] (5,0) circle (2pt);
\draw[fill=black] (7,0) circle (2pt);
\draw[fill=black] (6,1.73) circle (2pt);

\node at (6,-0.7) {$s$};

\draw[thick] (5,0) -- (7,0) -- (6,1.73);

\end{tikzpicture}

%% file: Figures/deanoymization_example.tex
\begin{tikzpicture}[node distance=4cm,auto,>=latex]
\draw[fill=black] (0,0) circle (2pt);
\draw[fill=black] (2,0) circle (2pt);
\draw[fill=black] (1,1.73) circle (2pt);

\node at (1,-0.7) {$s$};

\draw[thick] (0,0) -- (2,0) -- (1,1.73);

\draw[->, thick] (3,0.866) to (4,0.866);

\draw[fill=black] (5,0) circle (2pt);
\draw[fill=black] (7,0) circle (2pt);
\draw[fill=black] (6,1.73) circle (2pt);

\node at (4.6,0) {1};
\node at (7.4,0) {2};
\node at (6,2.1) {3};
\node at (6,-0.7) {$g_{1}$};

\draw[thick] (5,0) -- (7,0) -- (6,1.73);

\draw[fill=black] (8.5,0) circle (2pt);
\draw[fill=black] (10.5,0) circle (2pt);
\draw[fill=black] (9.5,1.73) circle (2pt);

\node at (8.1,0) {2};
\node at (10.9,0) {1};
\node at (9.5,2.1) {3};
\node at (9.5,-0.7) {$g_2$};

\draw[thick] (8.5,0) -- (10.5,0) -- (9.5,1.73);

\draw[fill=black] (12,0) circle (2pt);
\draw[fill=black] (14,0) circle (2pt);
\draw[fill=black] (13,1.73) circle (2pt);

\node at (11.6,0) {1};
\node at (14.4,0) {3};
\node at (13,2.1) {2};
\node at (13,-0.7) {$g_3$};

\draw[thick] (12,0) -- (14,0) -- (13,1.73);

\end{tikzpicture}

%% file: Figures/sys_model.tex
\tikzstyle{block}=[rectangle, draw, thick, minimum width=2em, minimum height=2em]

\begin{tikzpicture}[node distance=4cm,auto,>=latex']
    
    \node (dummy) {};
    \node [block, align=center] (enc) [left of=dummy, node distance=2.5cm] {Encoder\\$\phi_n$};
     \node [block, align=center] (dec) [right of=dummy, node distance=2.5cm] {Decoder\\$\psi_n$};
     \node (u) [left of=enc, node distance=2cm] {$\ua$};
     \node (u1) [below of=u, node distance=0.5cm] {$(\ga \text{ or } \sta)$};
     \node (hatu) [right of=dec, node distance=2cm] {$\hua$};
     \node (m) [above of=dummy, node distance=0.5cm] {$M_n \in \big[2^{{\binom{n}{2}R}}\big]$};
      \node (v) [above of=dec, node distance=1.75cm] {$\ub$};
      \node (v1) [right of=v, node distance=1.1cm] {$(\gb \text{ or } \stb)$};
     
     \draw[->, thick] (u.east) --  (enc.west);
     \draw[->, thick] (enc.east) --  (dec.west);
     \draw[->, thick] (dec.east) --  (hatu.west);
     \draw[->, thick] (v.south) --   (dec.north) ;
\end{tikzpicture}

%% file: Min_comp_rate.tex
\subsection{A Candidate for $R^*$}
If we would like to compress the labelled graph $\ua=\ga$ with the side information $\ub=\gb$ at the decoder, then we can consider the entries of the respective adjacency matrices $(A_{i,j}:1\leq i < j \leq n)$ and $(B_{i,j}:1\leq i < j \leq n)$, where the pairs $(A_{i,j}, B_{i,j})$ are i.i.d. with distribution $P_{\ea, \eb}$.  These are correlated strings of length $\binom{n}{2}$ with independence across the symbols. We can now use the classical source coding with side information result \cite{slepian_wolf} for memoryless sources to conclude that the minimum achievable compression rate is 
$$R^* = H(\ea | \eb ),$$
 which is equal to $\lim \limits_{n \to \infty} \frac{1}{\binom{n}{2}}H(\ga|\gb)$. 
 
On the other hand, if either of $\ua$ and $\ub$ is a structure, then we can no longer exploit the independence. Nevertheless, it is natural to expect that the minimum achievable compression rate should be 
\begin{align}
\lim \limits_{n \to \infty} \frac{1}{\binom{n}{2}}H(\ua|\ub) \label{eq:limit}
\end{align}
assuming that the limit exists. In the next section, we give a characterization of this limit.

\subsection{Characterization of the Limit in \eqref{eq:limit}}
Before we proceed to characterize the limit in \eqref{eq:limit}, we make a note of a few facts about the entropy \cite{choi_structural_entropy} of an Erd\H{o}s-R\'enyi (ER) random graph $G$ and its structure $S$. It is intuitively clear that most of the randomness in a graph must come from its structure as it retains the information about $\binom{n}{2}$ connections, while ignoring the information about $n$ vertex labels. The following  expressions precisely quantify this intuition. If $P_A$ stands for the probability distribution of an entry of the adjacency matrix under the ER model, then 
\begin{align}
    H(G) &= \binom{n}{2}H(A) \qquad \text{and} \label{eq:ent:g}
    \\ H(S) &= \binom{n}{2}H(A) -n\log n + \mathcal{O}(n). \label{eq:ent:struct}
\end{align}
The structural entropy \eqref{eq:ent:struct} of weighted graphs can easily be obtained by essentially following the arguments of \cite{choi_structural_entropy} used in the case of unweighted graphs. The key result used in \cite{choi_structural_entropy} to compute the structural entropy is that $|\operatorname{Aut}(G)|$ of an Erd\H{o}s-R\'enyi graph is $1$ (i.e, $G$ is asymmetric) with probability going to $1$ as $n \to \infty$.

\begin{lemma}\label{lem:asymp:limit}
Let $(\ga, \gb)$ be correlated Erd\H{o}s-R\'enyi random graphs specified by the distribution $P_{\ea, \eb}$. Let $\sta$ and $\stb$ be the respective structures of $\ga$ and $\gb$. Then
\begin{align}\label{eq:asymp:cond:struct}
    \lim_{n \to \infty}\frac{1}{\binom{n}{2}}H(\ga|\stb)= H(\ea | \eb )
\end{align}  
Moreover, $$\lim_{n \to \infty}\frac{1}{\binom{n}{2}}H(\ga|\gb)=\lim_{n \to \infty}\frac{1}{\binom{n}{2}}H(\sta|\gb)= \lim_{n \to \infty}\frac{1}{\binom{n}{2}}H(\sta|\stb)= H(\ea | \eb ).$$
\end{lemma}
\begin{proof}
Consider the relation $H(\ga|\gb) = H(\ga|\gb, \stb) =   H(\ga|\stb) - I(\ga ; \gb |\stb),$ and use the inequality $0 \leq I(\ga ; \gb |\stb) \leq H(\gb|\stb)$ to get 
\begin{align}\label{ineq:cnd}
     H(\ga|\gb) \leq H(\ga|\stb) \leq  H(\ga|\gb) + H(\gb|\stb).
\end{align}
Since the structure $\stb$ is a function of the graph $\gb$, we have 
\begin{align}
    H(\gb|\stb)= H(\gb)- H(\stb) &= \binom{n}{2}H(\eb)- \left[\binom{n}{2}H(\eb) -n\log n + \mathcal{O}(n)\right]\nonumber \\
    &= n\log n -\mathcal{O}(n),
\end{align} 
where, in the second equality,  we used \eqref{eq:ent:g} and \eqref{eq:ent:struct}.  Because of the  independence across the edges in the Erd\H{o}s-R\'enyi model, $H(\ga|\gb)= \sum_{i<j}H(A_{i,j}| B_{i,j})= \binom{n}{2}H(\ea | \eb )$. Thus, the inequalities in \eqref{ineq:cnd} become
$$  \binom{n}{2}H(\ea | \eb ) \leq H(\ga|\stb) \leq  \binom{n}{2}H(\ea | \eb ) + n\log n -O(n),$$
which readily yields the limit  \eqref{eq:asymp:cond:struct} as $\frac{n\log n}{\binom{n}{2}} \to 0$.  

The other limits in the lemma statement follow similarly by noting that
\begin{align*}
     &H(\ga|\gb)-H(\ga|\sta) \leq H(\sta|\gb) \leq  H(\ga|\gb) \quad \text{ and }\\
     &H(\sta|\gb) \leq H(\sta|\stb) \leq  H(\ga|\stb).
\end{align*}
\end{proof}
The work of  \cite{nikpey_graph_side} considered the graph compression problem in the case when $(\ga, \gb)$ are independently sampled from an underlying Erd\H{o}s-R\'enyi unweighted graph $G$ with parameter $p$. In particular, $\ga$ (or $\gb$) is obtained by sampling each edge in $G$ with probability $\gamma$. This is an instance of the correlated Erd\H{o}s-R\'enyi model with the distribution 
\begin{align*}
    P_{\ea, \eb}(0,0) &= (1-p)+p(1-\gamma)^2\\
    P_{\ea, \eb}(0,1) &= p\gamma(1-\gamma)\\
    P_{\ea, \eb}(1,0) &= p\gamma(1-\gamma)\\
    P_{\ea, \eb}(1,1) &=  p\gamma^2.
\end{align*}
It was shown in \cite{nikpey_graph_side} that 
\begin{align}
    R^* \leq  H(\ea)-\delta = h_2(p\gamma) - \delta, \label{ineq:nikpey:achievability}
\end{align}
where $\delta=2\frac{s^2\sigma^2}{(s+2)^2}$ with $\sigma^2= (p\gamma)^2(1-(p\gamma)^2)$ and $s = \min\left\{1, \frac{1-\epsilon-p}{(1-\epsilon+p)(1-(p\gamma)^2)}\right\}$ for a constant $0 < \epsilon <1-p$. On the other hand,
$$H(\ea | \eb) = p\gamma h_2(\gamma)+(1-p\gamma)h_2\left(\frac{p\gamma(1-\gamma)}{1-p\gamma}\right),$$
which is believed to be the minimum achievable compression rate $R^*$. It can be easily seen that  $H(\ea)-\delta$ is strictly larger than $H(\ea | \eb)$ for certain parameters; for instance, by setting $p=1/2$ and $\gamma=1$, we see that $H(\ea | \eb)=0$ and $H(\ea)-\delta=1- \frac{3}{8}\left(\frac{s}{s+2}\right)^2>0$.

\subsection{Main Results}
In order to prove our main result on what is the minimum achievable compression rate in the graph compression problem, we need the following auxiliary results. The first result is concerned with the asymptotic behavior of the maximum of $\pgagb(\ga, \pi(\gb))$ over all permutations $\pi$ of vertex labels under the CER model. We know  that if $\pi$ is the identity permutation, then by the weak law of large numbers,
\begin{align}\label{eq:asymp:conv:prob:noperm}
    \frac{1}{\binom{n}{2}} \log \frac{1}{ \pgagb(\ga, \gb)} &= \frac{1}{\binom{n}{2}} \sum_{i<j} \log \frac{1}{ P_{\ea, \eb}(\ea_{i,j}, \eb_{i,j})} \longrightarrow  H(\ea, \eb)
\end{align}
in probability, as $n \to \infty$. The following theorem shows that 
$$\frac{1}{\binom{n}{2}} \log \frac{1}{\max \limits_{\pi \in \mathcal{S}_n} \pgagb(\ga, \pi(\gb))}$$ and $$\frac{1}{\binom{n}{2}} \log \frac{1}{\max \limits_{\pi_{a}, \pi_{b} \in \mathcal{S}_n} \pgagb(\pi_{a}(\ga), \pi_{b}(\gb))}$$ also converge to the same limit $H(\ea, \eb)$, which means that the random variables $\max \limits_{\pi \in \mathcal{S}_n} \pgagb(\ga, \pi(\gb))$, $\max \limits_{\pi_{a}, \pi_{b} \in \mathcal{S}_n} \pgagb(\pi_{a}(\ga), \pi_{b}(\gb))$ and $\pgagb(\ga, \gb)$ are close by in the above sense.

\begin{theorem}\label{thm:asymp:prob:gen}
    For  a pair of correlated Erd\H{o}s-R\'enyi random graphs $(\ga,\gb)$,   
      \begin{align}\label{eq:asymp:conv:prob}
    \frac{1}{\binom{n}{2}} \log \frac{1}{\max \limits_{\pi \in \mathcal{S}_n} \pgagb(\ga, \pi(\gb))}  \longrightarrow  H(\ea, \eb)
\end{align}
and 
\begin{align}\label{eq:asymp2:conv:prob}
    \frac{1}{\binom{n}{2}} \log \frac{1}{\max \limits_{\pi_{a}, \pi_{b} \in \mathcal{S}_n} \pgagb(\pi_{a}(\ga), \pi_{b}(\gb))}  \longrightarrow  H(\ea, \eb)
\end{align}
in probability and almost surely as $n \to \infty$.
\end{theorem}
\begin{proof}
We will first prove \eqref{eq:asymp:conv:prob}. Fix an integer $n \geq 1$ and a permutation $\pi \in \mathcal{S}_n$ of the vertex labels. For a fixed $\delta>0$, let  $\mathcal{A}_{\pi}$ be the event that $\pgagb(\ga, \pi(\gb))>2^{\binom{n}{2}\delta}\pgagb(\ga, \gb)$. It follows that 
\begin{align}  
\mathbb{P}\left(\mathcal{A}_{\pi}\right)&= \sum_{\gga, \ggb}\pgagb\left(\gga, \ggb\right)\mathds{1}_{\mathcal{A}_{\pi}} \nonumber\\ &
< 2^{-\binom{n}{2} \delta} \sum_{\gga, \ggb} \pgagb(\gga, \pi(\ggb)),\nonumber
\end{align}
where the last inequality uses the definition of the event $\mathcal{A}_{\pi}$ and the fact that $0\leq 2^{-\binom{n}{2}\delta}\pgagb(\gga, \pi(\ggb))$ on $\mathcal{A}^c_{\pi}$. Since the map $\ggb \mapsto \pi(\ggb)$ is a bijective function, we have 
$$\sum_{\gga, \ggb}\pgagb(\gga, \pi(\ggb)) = \sum_{\gga, \pi^{-1}(\ggb)}\pgagb(\gga, \ggb)= 1, $$ and therefore, 
\begin{align}
\mathbb{P}\left(\mathcal{A}_{\pi}\right) &< 2^{-\binom{n}{2}\delta}. \label{eq:perm:error}
\end{align}

Let $\mathcal{A}$ be the event that $\max_{\pi \in \mathcal{S}_n}\pgagb(\ga, \pi(\gb))>2^{\binom{n}{2}\delta}\pgagb(\ga, \gb)$, which is $\cup_{\pi \in \mathcal{S}_n}\mathcal{A}_{\pi}$. Then by the union bound, $    \mathbb{P}\left(\mathcal{A}\right) \leq \sum_{\pi \in \mathcal{S}_n}\mathbb{P} \left(\mathcal{A}_{\pi}\right)$. By using \eqref{eq:perm:error} and the fact that  there are only $n!$ permutations, we obtain
\begin{align} \label{eq:bound:conv:prob}    \mathbb{P}\left(\mathcal{A}\right) < n! \cdot 2^{-\binom{n}{2}\delta}  \leq  2^{n \log n} \cdot 2^{-\binom{n}{2}\delta}.
\end{align}
Let $\mathcal{B}$ be the event  
$$\left\lvert\frac{1}{\binom{n}{2}} \log  \frac{1}{\max \limits_{\pi \in \mathcal{S}_n} \pgagb(\ga, \pi(\gb))} -  H(\ea, \eb)\right\rvert>2\delta,$$
and ${\mathcal{B}}'$ be the event  
$$\left\lvert\frac{1}{\binom{n}{2}} \log  \frac{1}{\pgagb(\ga, \gb)} -  H(\ea, \eb)\right\rvert>\delta.$$
Since $\pgagb(\ga, \gb) \leq \max \limits_{\pi \in \mathcal{S}_n}\pgagb(\ga, \pi(\gb)) \leq 2^{\binom{n}{2}\delta}\pgagb(\ga, \gb)$ on $\mathcal{A}^c$, it can be easily verified that $\mathcal{B}\cap \mathcal{A}^c \subseteq \mathcal{B}' \cap \mathcal{A}^c$. This implies that
\begin{align}
     \mathbb{P}\left(\mathcal{B}\cap \mathcal{A}^c\right) 
  \leq \mathbb{P}\left(\mathcal{B}'\cap \mathcal{A}^c\right)
    & \leq \mathbb{P}\left(\left\lvert\frac{1}{\binom{n}{2}} \log  \frac{1}{ \pgagb(\ga, \gb)} -  H(\ea, \eb)\right\rvert>\delta \right)\nonumber\\
    &  \leq 2^{-\frac{2 \delta^2 \binom{n}{2}}{(\log C)^2}}, \label{eq:bound:conv:prob2}
\end{align}
where $C$ is the smallest non-zero value of the joint distribution $P_{\ea, \eb}$, 
 and the last inequality follows by applying Hoeffding's inequality \cite[Theorem~2.8]{boucheron_book} to the sum $\frac{1}{\binom{n}{2}} \sum_{i<j} \log \frac{1}{ P_{\ea, \eb}(\ea_{i,j}, \eb_{i,j})}$ consisting of i.i.d random variables with $\log \frac{1}{ P_{\ea, \eb}(\ea_{i,j}, \eb_{i,j})} \in \left[0,\log \frac{1}{C}\right]$ with probability 1.
Therefore, by combining  \eqref{eq:bound:conv:prob} and \eqref{eq:bound:conv:prob2}, we obtain
\begin{align}
\mathbb{P}\left(\mathcal{B}\right) 
&\leq  \mathbb{P}\left( \mathcal{A}\right) + \mathbb{P}\left(\mathcal{B}\cap \mathcal{A}^c\right) 
\leq 2^{-\left[\binom{n}{2}\delta- n \log n\right]} + 2^{-\frac{2 \delta^2 \binom{n}{2} }{(\log C)^2}}, \label{eq:bound:conv:prob3}
\end{align}
proving the convergence of \eqref{eq:asymp:conv:prob} in probability. As the bound in \eqref{eq:bound:conv:prob3} is summable, the Borel-Cantelli lemma implies the almost sure convergence of \eqref{eq:asymp:conv:prob}.

For the proof of \eqref{eq:asymp2:conv:prob}, observe that 
$$\max \limits_{\pi_{a}, \pi_{b} \in \mathcal{S}_n} \pgagb(\pi_{a}(\ga), \pi_{b}(\gb)) = \max \limits_{\pi \in \mathcal{S}_n} \pgagb(\ga, \pi(\gb)).$$
This is simply because for $\pi_{a}, \pi_{b} \in \mathcal{S}_n$, 
\begin{align}
    \pgagb(\pi_{a}(\ga), \pi_{b}(\gb))
    &= \prod_{i<j} P_{\ea, \eb}(\ea_{\pi_{a}^{-1}(i), \pi_{a}^{-1}(j)}, \eb_{\pi_{b}^{-1}(i), \pi_{b}^{-1}(j)})\nonumber\\
    & = \prod_{i'<j'} P_{\ea, \eb}(\ea_{i',j'}, \eb_{\pi_{b}^{-1}\circ \pi_{a}(i'), \pi^{-1}_b\circ \pi_{a}(j')})\label{eq:symm}\\
    &=\pgagb(\ga, \pi^{-1}_a \circ\pi_{b}(\gb)),\nonumber
\end{align}
where the equality \eqref{eq:symm} uses the symmetry of the adjacency matrices.
\end{proof}

 The above result will be helpful in handling the probabilities involving structures. The next theorem, which concerns the asymptotic behaviour of the conditional probabilities involving structures, is a key ingredient in showing that $H(A|B)$ is indeed the minimum achievable compression rate $R^*$.

\begin{theorem}\label{thm:asymp:inprob:weighted}
Let $(\ga,\gb)$ be a pair of correlated Erd\H{o}s-R\'enyi random graphs. Let $\sta$ and $\stb$ be the structures of the graphs $\ga$ and $\gb$, respectively. Then,
    \begin{align} \label{eq:asymp:inprob:weighted}
     \frac{1}{\binom{n}{2}}\log\frac{1}{\pgacsb(\ga| \stb)}  \longrightarrow  H(\ea | \eb )   
\end{align}
in probability as $n \to \infty$. In addition, the random variables $\frac{1}{\binom{n}{2}}\log\frac{1}{\pgacgb(\ga| \gb)}, \frac{1}{\binom{n}{2}}\log \frac{1}{\psacgb(\sta| \gb)}$ and  $\frac{1}{\binom{n}{2}}\log\frac{1}{\psacsb(\sta| \stb)}$ also converge in probability to $H(\ea | \eb )$.
\end{theorem}
\begin{proof}
We will first show that \eqref{eq:asymp:inprob:weighted} holds. Since the random structure $\stb$ is a function of the random graph $\gb$, and the probability of a structure is equal to the sum of the probabilities of distinct graphs having that structure, the following relations hold with probability (w.p.) $1$:
$$\psb(\stb) = \frac{1}{|\operatorname{Aut}(\gb)|}\sum \limits_{\pi \in \mathcal{S}_n} \pgb(\pi(\gb))\quad \text{and}$$
$$\pgasb(\ga, \stb) = \frac{1}{|\operatorname{Aut}(\gb)|}\sum \limits_{\pi \in \mathcal{S}_n} \pgagb(\ga, \pi(\gb)).$$
In the above expressions, the summations are over the set of all permutations of $n$ vertex labels. So, a distinct graph may appear multiple times in the summations. In particular, the number of terms in each summation corresponding to a distinct graph is exactly $|\operatorname{Aut}(\gb)|$ number of distinct permutations. However, the  prefactor $\frac{1}{|\operatorname{Aut}(\gb)|}$ nullifies this overcounting in the summations to get the desired probability. 

Using the above expressions, the conditional probability distribution function  can be expressed as 
    \begin{align}
        \pgacsb(\ga| \stb)&= \frac{\pgasb(\ga, \stb)}{\psb(\stb)}\nonumber\\
        & =\frac{\frac{1}{|\operatorname{Aut}(\gb)|}\sum \limits_{\pi \in \mathcal{S}_n} \pgagb(\ga, \pi(\gb))}{\frac{1}{|\operatorname{Aut}(\gb)|}\sum \limits_{\pi \in \mathcal{S}_n} \pgb(\pi(\gb))} \nonumber\\
        &= \frac{\sum \limits_{\pi \in \mathcal{S}_n} \pgagb(\ga, \pi(\gb))}{\sum \limits_{\pi \in \mathcal{S}_n} \pgb(\pi(\gb))},
    \end{align}
which holds w.p. $1$. Thus
\begin{align}\label{eq:cond:expansion:general}
     \frac{1}{\binom{n}{2}} \log  \frac{1}{\pgacsb(\ga| \stb)} \nonumber &=  \frac{1}{\binom{n}{2}} \log \frac{1}{\sum \limits_{\pi \in \mathcal{S}_n} \pgagb(\ga, \pi(\gb))} 
     -  \frac{1}{\binom{n}{2}} \log \frac{1}{\sum \limits_{\pi \in \mathcal{S}_n}\pgb(\pi(\gb))}.
\end{align}
 If we can show that the first term  on the right-hand side of \eqref{eq:cond:expansion:general} converges in probability to $H(\ea, \eb)$, and the second term converges in probability to $H(\eb)$, then by the linearity of convergence in probability, we will have the result \eqref{eq:asymp:inprob:weighted}. 

Consider the second term on the right-hand side of \eqref{eq:cond:expansion:general}. As the probability 
$$\pgb(\pi(\gb)) = \prod_{i<j} P_{\eb}( \eb_{\pi^{-1}(i),\pi^{-1}(j)})=\prod_{i'<j'} P_{\eb}( \eb_{i',j'})$$
is invariant under any permutation $\pi$, $\sum_{\pi \in \mathcal{S}_n}\pgb(\pi(\gb)) = n! \pgb(\gb)$. Hence, by noting that $\frac{\log {n!}}{\binom{n}{2}} \to 0$ and by applying the weak law of large numbers, we have
\begin{align}
    \frac{1}{\binom{n}{2}} \log \frac{1}{n! \pgb(\gb)} & =  \frac{1}{\binom{n}{2}} \log \frac{1}{\pgb(\gb)} - \frac{\log {n!}}{\binom{n}{2}} \nonumber\\
    & =  \frac{1}{\binom{n}{2}} \sum_{i<j} \log \frac{1}{ P_{\eb}(\eb_{i,j})}  - \frac{\log {n!}}{\binom{n}{2}} \nonumber\\
    & \longrightarrow  H(\eb)
\end{align}
in probability, as $n\to\infty$. 

Let us now argue the convergence of the first term on the right-hand side of \eqref{eq:cond:expansion:general} to $H(\ea, \eb)$.  As 
\begin{align}
    \log \frac{1}{n! \max \limits_{\pi \in \mathcal{S}_n} \pgagb(\ga, \pi(\gb))} &\leq \log \frac{1}{\sum \limits_{\pi \in \mathcal{S}_n}\pgagb(\ga, \pi(\gb))}
    \leq \log \frac{1}{\max \limits_{\pi \in \mathcal{S}_n} \pgagb(\ga, \pi(\gb))}, \nonumber
\end{align}
and $\frac{\log {n!}}{\binom{n}{2}} \to 0$, the desired convergence of the first term on the right-hand side of \eqref{eq:cond:expansion:general} is equivalent to  
    \begin{align}
    \frac{1}{\binom{n}{2}} \log \frac{1}{\max \limits_{\pi \in \mathcal{S}_n} \pgagb(\ga, \pi(\gb))}  \longrightarrow  H(\ea, \eb)
\end{align}
in probability, which is guaranteed by Theorem~\ref{thm:asymp:prob:gen}. This completes the proof of \eqref{eq:asymp:inprob:weighted}.

We know that by the independence across edges and the weak law of large numbers,
$
    \frac{1}{\binom{n}{2}} \log \frac{1}{ \pgagb(\ga|\gb)} $ converges in probability to $ H(\ea| \eb)
$
as $n \to \infty$. The rest of the theorem can be proved analogously by noting that 
\begin{align*}
        \psacgb(\sta| \gb)&= \frac{\sum \limits_{\pi \in \mathcal{S}_n} \pgagb(\pi(\ga), \gb)}{ |\operatorname{Aut}(\ga)|\pgb(\gb)},
    \end{align*}
    \begin{align*}
        \psacsb(\sta| \stb)&= \frac{\sum \limits_{\pi_{a}, \pi_{b} \in \mathcal{S}_n} \pgagb(\pi_{a}(\ga), \pi_{b}(\gb))}{ {|\operatorname{Aut}(\ga)|}\sum \limits_{\pi \in \mathcal{S}_n} \pgb(\pi(\gb))},
    \end{align*}
w.p. $1$ and $1 \leq |\operatorname{Aut}(\ga)| \leq n!$,
and by applying Theorem~\ref{thm:asymp:inprob:weighted} appropriately.
\end{proof}

We are now ready to state our main result for compression of $\ua$, which is $\ga$ or $\sta$, with $\ub$, which is $\gb$ or $\stb$, as side information to the decoder.
\begin{theorem}\label{thm:compress:binary:cer}
    The minimum achievable compression rate is
    \[
     R^* = H(\ea | \eb ).
    \]
    
\end{theorem}
\begin{proof}
\noindent\underline{Converse part:} Consider any sequence of encoding and decoding schemes $\{(\phi_n,\psi_n)\}_{n=1}^{\infty}$ such that $M_n=\phi_n(\ua) \in \Big[2^{^{\binom{n}{2}R}}\Big]$ is the message and  $\hua=\psi_n\left(M_n, \ub\right)$ is the recovered graph with the probability of error going to zero, i.e.,
    $$\epsilon_n\triangleq\mathbb{P}\left\{\hua \neq \ua\right\} \to 0.$$
    By Fano's inequality, we have $H(\ua |M_n, \ub) \leq H(\ua |\hua)\leq h_2(\epsilon_n)+\epsilon_n \log 2^{{\binom{n}{2}}} \leq  1+ \binom{n}{2} \epsilon_n,$
    where $h_2(x):= -x \log x-(1-x)\log(1-x), x \in [0,1]$, with the convention that $0 \log 0:=0$.
The rate $R$ can be bounded from below as follows:
\begin{align*}
    \binom{n}{2}R=\log |\mathcal{M}_n|  &\geq H(M_n) \\
    & \geq H(M_n|\ub)\\
    & \stackrel{(a)}{=} I(M_n;\ua|\ub)\\
    & = H(\ua|\ub) -H(\ua|M_n, \ub) \\
    & \stackrel{(b)}{\geq} H(\ua|\ub) - 1- \binom{n}{2} \epsilon_n,
\end{align*}
where $(a)$ uses the fact that $M_n$ is a function of $\ua$, and $(b)$ is a consequence of Fano's inequality. By letting $n$ approach infinity, we get
\begin{align}\label{ineq:conv:bound1}
    R \geq \lim_{n \to \infty}\frac{1}{\binom{n}{2}}H(\ua|\ub).
\end{align}
We know from Lemma~\ref{lem:asymp:limit} that when $\ua$ is $\ga$ or $\sta$, and $\ub$  is $\gb$ or $\stb$, $\lim_{n \to \infty}\frac{1}{\binom{n}{2}}H(\ua|\ub)=H(\ea | \eb)$. As \eqref{ineq:conv:bound1} holds for any achievable rate $R$,  we have the converse result  
\begin{align}
    R^* \geq H(\ea | \eb ).
\end{align}

\noindent\underline{Achievability part:} Let us now show that any rate arbitrarily close to $H(\ea | \eb )$ is achievable, which implies $R^* \leq H(\ea | \eb )$. We use the standard random binning approach to show the existence of an encoding and decoding scheme with the desired performance. Generate a codebook by assigning each realization $\uua \in \big[2^{{\binom{n}{2}}}\big]$ to an index $\Phi_n(\uua) \in \big[2^{{\binom{n}{2}R}}\big]$ chosen uniformly at random. Assume that the codebook chosen in this fashion is made available to the decoder as well. Fix a $\delta>0$, and let 
\begin{align}
    \mathcal{T}_n(\delta) \triangleq \left\{ (\uua, \uub) : \left \lvert \frac{1}{\binom{n}{2}}\log\frac{1}{\puacub(\uua| \uub)} - H(\ea | \eb )\right \rvert \leq \delta \right\}.
\end{align}
Upon seeing $\ua$, the encoder outputs the corresponding index $\Phi_n(\ua)$. Given an index $I \in \big[2^{{\binom{n}{2}R}}\big]$ and side information $\ub$, the decoder looks for a unique $\huua$ such that $I=\Phi_n(\huua)$ and $(\huua, \ub) \in \mathcal{T}_n(\delta)$. 

We will now argue that if $R = H(\ea | \eb ) + 2\delta$, then the probability of decoding error averaged over all random codebooks $ \Bar{\epsilon}_n\triangleq\mathbb{P}\big(\hua \neq \ua\big)$ goes to zero; this will then imply the existence of a codebook with the probability of decoding error $\epsilon_n \to 0$. For the encoding and decoding schemes we considered, it is well-known  \cite[Lemma~7.2.1.]{han_book} that 
\begin{align}\label{eq:errorbound}
\Bar{\epsilon}_n \leq \mathbb{P}(\mathcal{T}^{^c}_n(\delta))+2^{-\binom{n}{2}\delta}.
\end{align}
For the sake of completeness, we will give its proof here. Define the error events $\mathcal{E}_1:=\left\{(\ua, \ub) \notin \mathcal{T}_n(\delta)\right\}$ and $\mathcal{E}_2:=\left\{\exists \  \huua \neq \ua \text{ s.t. } \Phi_n(\huua)=\Phi_n(\ua), (\huua, \ub) \in \mathcal{T}_n(\delta)\right \}$, where $\Phi_n$ is a random function.  For a decoding error, one of the error events must occur. Therefore, by the union bound, we have
\begin{align*}
    \Bar{\epsilon}_n \leq \mathbb{P}(\mathcal{E}_1) + \mathbb{P}( \mathcal{E}_2) = \mathbb{P}(\mathcal{T}^{^c}_n(\delta)) + \mathbb{P}(\mathcal{E}_2).
\end{align*}
It is enough to show that $\mathbb{P}(\mathcal{E}_2) \leq 2^{-\binom{n}{2}\delta}$.  
\begin{align}
    \mathbb{P}(\mathcal{E}_2) & \leq \sum_{ \uua, \uub} \bigg[\puaub(\uua, \uub)\cdot\sum_{\substack{\huua \neq \uua:\\ (\huua, \uub) \in \mathcal{T}_n(\delta)}} \mathbb{P}(\Phi_n(\huua) = \Phi_n(\uua))\bigg]\nonumber\\
    & = \sum_{ \uua, \uub} \bigg[\puaub(\uua, \uub)\cdot\sum_{\substack{\huua \neq \uua: \\ (\huua, \uub) \in \mathcal{T}_n(\delta)}} \frac{1}{2^{^{\binom{n}{2}R}}} \bigg]\nonumber\\
    & \leq \frac{1}{2^{^{\binom{n}{2}R}}} \sum_{ \uua, \uub} \puaub(\uua, \uub) \left|\left\lbrace \huua :(\huua, \uub) \in \mathcal{T}_n(\delta) \right \rbrace \right|. \label{eq:card:bound}
\end{align}
Since $\puacub(\huua| \uub) \geq 2^{-H(\ea | \eb) - \delta}$ for any $(\huua, \uub) \in \mathcal{T}_n(\delta)$, we have 
\begin{align}
1 &\geq \sum_{ \huua:(\huua, \uub) \in \mathcal{T}_n(\delta)} \puacub(\huua| \uub) \\& \geq 2^{-\binom{n}{2}\left[H(\ea | \eb) + \delta\right]} \left|\left\lbrace \huua :(\huua, \uub) \in \mathcal{T}_n(\delta) \right \rbrace \right|, 
\end{align}
which implies that $\left|\left\lbrace \huua :(\huua, \uub) \in \mathcal{T}_n(\delta) \right \rbrace \right| \leq 2^{\binom{n}{2}\left[H(\ea | \eb) + \delta\right]}$. By plugging this into \eqref{eq:card:bound}, we have that for $R = H(\ea | \eb ) + 2\delta$,
\begin{align*}
    \mathbb{P}(\mathcal{E}_2) &\leq 2^{\binom{n}{2}\left[H(\ea | \eb) - R + \delta\right]}\sum_{ \uua, \uub} \puaub(\uua, \uub)\\
    & = 2^{-\binom{n}{2}\delta},
\end{align*}
yielding the inequality \eqref{eq:errorbound}.

We know from  Theorem~\ref{thm:asymp:inprob:weighted} that 
$$\frac{1}{\binom{n}{2}}\log\frac{1}{\puacub(\ua| \ub)}  \longrightarrow  H(\ea | \eb ),$$
in probability, which means that for any $\delta >0$, 
$\mathbb{P}(\mathcal{T}^{^c}_n(\delta)) \to 0$ (or, equivalently $\mathbb{P}(\mathcal{T}_n(\delta)) \to 1$) as $n \to \infty$. This implies that the right-hand side of \eqref{eq:errorbound} goes to zero as $n \to \infty$, i.e.,  the probability of error (averaged over codebooks) $\Bar{\epsilon}_n \to 0$. Hence, there exists a codebook with the compression rate $H(\ea | \eb ) + 2\delta$. Since $\delta >0$ is arbitrary, we have $R^*\leq H(\ea | \eb )$, completing the proof of the theorem.
\end{proof}

It follows from the proof of Theorem~\ref{thm:asymp:inprob:weighted} that the decoder can label the involved structure in all possible ways and check if the maximum of the joint probability distribution over all labelings satisfies a typicality condition. For instance, consider the case  $\ua=\ga$ and $\ub=\stb$. It follows from the proof of Theorem~\ref{thm:asymp:inprob:weighted} that in order to check if a graph $\gga$ in a bin and the side information $\stbrel$ satisfy the condition $(\gga,\stbrel) \in \mathcal{T}_n(\delta)$, the decoder can first verify if the side information $\stbrel$ is a typical unlabelled graph or not, and then check the joint typicality condition
\begin{align}
    2^{-\binom{n}{2}\left[H(\ea,\eb) +\frac{\delta}{2}\right]} \leq \max_{\ggb' \cong \stbrel}\pgagb(\gga, \ggb') \leq 2^{-\binom{n}{2}\left[H(\ea,\eb) -\frac{\delta}{2}\right]}
\end{align}
by labelling $\stbrel$ in all possible ways.

%% file: Unweighted_graphs.tex
In this section, we consider the special case of unweighted graphs. When the graphs are unweighted, the entries of the adjacency matrices take values in $\{0,1\}$. In this case, we can cast Theorem~\ref{thm:asymp:prob:gen} and the joint typicality condition in terms of graph alignment. 

\subsection{Graph Alignment}
Let us briefly look at the idea of graph alignment, which finds its applications in deanonymization \cite{pedarsani_2011_privacy, narayan_deanonymize} and pattern recognition \cite{berg_pattern}. In a nutshell, graph alignment involves finding a correspondence between the vertex sets of two graphs while optimizing an objective function. Let  $(\ga,\gb)$ be the
correlated Erd\H{o}s-R\'enyi random (unweighted) graphs with $P_{\ea, \eb}$ being the corresponding joint distribution of the entries of the adjacency matrices. We use the shorthand notation $\pab$ instead of $P_{\ea, \eb}(a,b)$ for $a, b \in \{0,1\}$.

For example, let us consider the graph deanonymization problem. Here, a structure $\stbrel$ corresponding to a true graph $\ggb$ is revealed publicly. The aim is to identify the true graph from $\stbrel$ using another correlated labelled  graph $\gga$. Given $\stbrel$ alone, all the distinct graphs obtained by labelling $\stbrel$ in all possible ways are equally likely. Note that the true graph $\ggb$ corresponds to some labelling of $\stbrel$. However, if one has access to a correlated labelled graph $\gga$, then they can use this extra information to estimate the true graph of $\stbrel$ using $\gga$. This can be done by using the  MAP estimator, which is optimal in the sense of minimizing the probability of error in recovering the true graph $\ggb$. The MAP estimator is given by
\begin{align}
&\underset{\ggb'}{\operatorname{argmax}} \  P_{\scriptscriptstyle \gb|\ga,\stb}\left(\ggb'|\gga,\stbrel\right) =  \underset{\ggb' \cong \stbrel}{\operatorname{argmax}} \  P_{\scriptscriptstyle \ga,\gb}\left(\gga, \ggb'\right) \nonumber. 
\end{align}
The joint distribution $P_{\scriptscriptstyle \ga,\gb}\left(\gga, \ggb'\right)$ in the unweighted case can be expressed as
\begin{align}
     &\prod_{i<j} \pzerozero^{(1-{a}_{i,j})(1-b'_{i, j})}\ \ponezero^{a_{i,j}(1-b'_{i, j})}\  \pzeroone^{(1-{a}_{i,j})b'_{i,j}}  \ \poneone^{a_{i,j}b'_{i,j}}\nonumber\\
     & \qquad \qquad=\pzerozero^{\binom{n}{2}}
    \left(\frac{\ponezero}{\pzerozero}\right)^{\sum \limits_{i<j}a_{i,j}}
    \left(\frac{\pzeroone}{\pzerozero}\right)^{\sum \limits_{i<j}b'_{i,j}}
    \left(\frac{\poneone\pzerozero}{\ponezero\pzeroone}\right)^{\sum\limits_{i<j}a_{i,j}b'_{i,j}},\label{eq: unweighted:joint}
\end{align}
where $a_{i,j}  \in \{0,1\}$ and $b_{i,j} \in \{0,1\}$, $i, j \in [n]$, are the entries of the adjacency matrices of $\gga$ and $\ggb$, respectively. Since the values of the sums $\sum \limits_{i<j}a_{i,j}$ and $\sum \limits_{i<j}b'_{i,j}$ remain unchanged for any $\ggb' \cong \stbrel$, the MAP estimator becomes 
\begin{align*}
\underset{\ggb' \cong \stbrel}{\operatorname{argmax}} \left(\frac{\poneone\pzerozero}{\ponezero\pzeroone}\right)^{\sum_{i<j}a_{i,j}b'_{i,j}}.
\end{align*}
By using the fact that all the graphs having the structure $\stbrel$ can be expressed as relabelled versions of the true graph $\ggb$, the MAP estimator can be  written equivalently as 
\begin{align*}
\underset{\pi^{-1}(\ggb): \pi  \in \mathcal{S}_n}{\operatorname{argmax}} \left(\frac{\poneone\pzerozero}{\ponezero\pzeroone}\right)^{\sum_{i<j}a_{i,j}b_{\pi(i),\pi(j)}}.
\end{align*}
 Thus, if $\poneone\pzerozero>\ponezero\pzeroone$, the MAP estimator chooses a graph $\ggb'= \pi^{-1}(\ggb)$ for some $\pi \in \mathcal{S}_n$ such that it has the maximum value of $\sum_{i<j}a_{i,j}b'_{i,j}=\sum_{i<j}a_{i,j}b_{\pi(i), \pi(j)}$. In the other case of $\poneone\pzerozero<\ponezero\pzeroone$, it produces the one that minimizes $\sum_{i<j}a_{i,j}b'_{i,j}=\sum_{i<j}a_{i,j}b_{\pi(i), \pi(j)}$. 

 The graph alignment statistic $\sum_{i<j}\ea_{i,j}\eb_{i, j}$ counts the pairs of vertices which have an edge in both the graphs $\ga$ and $\gb$. Given a graph $\ga$ and a structure $\stb$, the MAP estimator identifies the labelling of $\stb$ that optimizes this statistic. We can also  consider other equivalent statistics instead of $\sum_{i<j}\ea_{i,j}\eb_{i, j}$; for instance, the so-called \emph{matching error} \cite{pedarsani_2011_privacy}, $\sum_{i<j}\left[(1-\ea_{i,j})\eb_{i, j} + \ea_{i,j}(1-\eb_{i,j})\right]$, can also be used for the MAP estimation. The matching error measures the number of pairs of vertices for which the vertices are connected by an edge in $\ga$ but not in $\gb$, or the vertices are connected by an edge in $\gb$ but not in $\ga$.

 It was shown in \cite{cullina_improved, cullina_exact} that the MAP estimator recovers the true graph with probability going to one as the number vertices grow larger and larger. This result is stated next. Though the original result applies to a more general setting where $P_{\ea, \eb}$ depends on $n$, we state it below in the special case of constant $P_{\ea, \eb}$.
 
 \begin{theorem}[\cite{cullina_exact,cullina_improved}]\label{thm:cullina}
       For correlated Erd\H{o}s-R\'enyi random (unweighted) graphs $(\ga,\gb)$ with  $\poneone\pzerozero > \ponezero\pzeroone$    
    \begin{align}\label{eq:max:close}
    \mathbb{P}\left(\sum_{i<j}\ea_{i,j}\eb_{i,j} < \max_{\pi \in \mathcal{S}_n}\sum_{i<j}\ea_{i,j}\eb_{\pi(i),\pi(j)} \right) \longrightarrow 0
\end{align}
as $n \to \infty$. If $\poneone\pzerozero < \ponezero\pzeroone$, then
\begin{align}\label{eq:min:close}
    \mathbb{P}\left(\sum_{i<j}\ea_{i,j}\eb_{i,j} > \min_{\pi \in \mathcal{S}_n}\sum_{i<j}\ea_{i,j}\eb_{\pi(i),\pi(j)} \right) \longrightarrow 0
\end{align} 
as $n \to \infty$.
\end{theorem}
 
The convergence in Theorem~\ref{thm:cullina} is in fact exponential in $n$. More precisely, the probabilities are bounded above by $2^{-[(n-2)(\sqrt{\pzerozero\poneone}-\sqrt{\pzeroone\ponezero})^2- 2\log n]}$. Since we know that 
\begin{align*}
    \min \limits_{\pi \in \mathcal{S}_n}\sum_{i<j}\ea_{i,j}\eb_{\pi(i),\pi(j)} &\leq \sum_{i<j}\ea_{i,j}\eb_{i,j}\leq \max \limits_{\pi \in \mathcal{S}_n}\sum_{i<j}\ea_{i,j}\eb_{\pi(i),\pi(j)},
\end{align*}
Theorem~\ref{thm:cullina} says that $\sum_{i<j}\ea_{i,j}\eb_{i,j}$ is equal to $\max \limits_{\pi \in \mathcal{S}_n}\sum_{i<j}\ea_{i,j}\eb_{\pi(i),\pi(j)}$ (resp. $\min \limits_{\pi \in \mathcal{S}_n}\sum_{i<j}\ea_{i,j}\eb_{\pi(i),\pi(j)}$) with high probability when $\poneone\pzerozero > \ponezero\pzeroone$ (resp. $\poneone\pzerozero < \ponezero\pzeroone$). This means the output of the MAP estimator is the true graph $\ggb$ with high probability.

\subsection{Unweighted Graph Compression with Side Information}

In the (unweighted) graph compression with side information, we can establish a convergence result, closely related to Theorem~1,  for the graph alignment statistic. This result, which is stated next as a lemma, is an immediate consequence of Theorem~\ref{thm:cullina}.

\begin{lemma}\label{lem:asymp:qap}
    For correlated Erd\H{o}s-R\'enyi random (unweighted) graphs $(\ga,\gb)$ with  $\poneone\pzerozero > \ponezero\pzeroone$    
    \begin{align}\label{eq:max:conv:prob}
    \frac{\max \limits_{\pi \in \mathcal{S}_n}\sum_{i<j}\ea_{i,j}\eb_{\pi(i),\pi(j)}}{\binom{n}{2}} \longrightarrow \poneone
\end{align}
in probability as $n \to \infty$. On the other hand, if $\poneone\pzerozero < \ponezero\pzeroone$ then
\begin{align}\label{eq:min:conv:prob}
    \frac{\min \limits_{\pi \in \mathcal{S}_n}\sum_{i<j}\ea_{i,j}\eb_{\pi(i),\pi(j)}}{\binom{n}{2}} \longrightarrow \poneone
\end{align}
in probability as $n \to \infty$. The same statements hold even for the maximum and minimum of  \linebreak $\sum \limits_{i<j}\ea_{\pi_{a}(i),\pi_{a}(j)}\eb_{\pi_{b}(i),\pi_{b}(j)}$ over $\pi_{a}, \pi_{b} \in \mathcal{S}_n$.
\end{lemma}
\begin{proof}
    Let $\poneone\pzerozero > \ponezero\pzeroone$. Fix a $\delta >0$,
\begin{align}  &\mathbb{P}\left(\left\lvert\frac{\max_{\pi \in \mathcal{S}_n}\sum_{i<j}\ea_{i,j}\eb_{\pi(i),\pi(j)}}{\binom{n}{2}} - p_{11}\right\rvert>\delta\right)\nonumber \\ & \quad = \mathbb{P}\left(\left\lvert\max_{\pi \in \mathcal{S}_n}\sum_{i<j}\ea_{i,j}\eb_{\pi(i),\pi(j)}- \sum_{i<j}\ea_{i,j}\eb_{i,j} \right\rvert>{\binom{n}{2}}\frac{\delta}{2}\right)
+\mathbb{P}\left(\left\lvert\frac{\sum_{i<j}\ea_{i,j}\eb_{i,j} }{\binom{n}{2}} - p_{11}\right\rvert>\frac{\delta}{2}\right) \label{eq:firstterm}
\end{align}
As we know that 
\begin{align}
    \frac{\sum_{i<j}\ea_{i,j}\eb_{i,j}}{\binom{n}{2}} \longrightarrow \mathbb{E}[\ea\eb]=p_{11}\nonumber
\end{align}
in probability by the weak law of large numbers, the second term in \eqref{eq:firstterm} goes to zero. The first term also goes to zero because of Theorem~\ref{thm:cullina}. This completes the proof of \eqref{eq:max:conv:prob}. The rest of the theorem statements follow in a similar way.
\end{proof}

Though the statement of Lemma~\ref{lem:asymp:qap} is weaker than that of Theorem~\ref{thm:cullina}, we record it here as it could be of independent interest. 

When trying to check if an $\uua$ in a bin and the side information $\uub$ satisfy $(\uua,\uub) \in \mathcal{T}_n(\delta)$, the decoder can use the graph alignment statistic.
Because
\begin{align}
    \frac{-\log \pgagb(\gga, \ggb')}{\binom{n}{2}} \nonumber= -\log \pzerozero - \frac{\sum_{i<j}\ea_{i,j}}{\binom{n}{2}}\log \frac{\ponezero}{\pzerozero} -\frac{\sum_{i<j}\eb'_{i,j}}{\binom{n}{2}}\log  \frac{\pzeroone}{\pzerozero} 
    - \frac{\sum\limits_{i<j}\ea_{i,j}\eb'_{i, j}}{\binom{n}{2}}\log \frac{\poneone\pzerozero}{\ponezero\pzeroone}, \nonumber
\end{align}
the decoder can, for instance, in the case of  $\ua=\ga$ and $\ub=\stb$, verify $(\gga,\stbrel) \in \mathcal{T}_n(\delta)$ by first checking if the graph $\gga$ in the bin and the structure $\stbrel$ are typical or not, i.e.,
$$\binom{n}{2}\left(\ponezero+ \poneone - \delta_1\right)\leq \sum_{i<j}\ea_{i,j} \leq \binom{n}{2}\left(\ponezero+ \poneone + \delta_1\right), $$
and 
$$\binom{n}{2}\left(\pzeroone+ \poneone - \delta_2\right)\leq \sum_{i<j}\eb'_{i,j} \leq \binom{n}{2}\left(\pzeroone+ \poneone + \delta_2\right)$$
where $\delta_1:=\delta/ \lvert \log \frac{\ponezero}{\pzerozero}\rvert$ and $\delta_2:=\delta/ \lvert \log \frac{\pzeroone}{\pzerozero}\rvert$. Here $\sum_{i<j}\eb'_{i,j}$ counts the number of edges in the structure $\stbrel$. Next, the decoder can label
$\stbrel$ in all possible ways and check if one of the following graph alignment conditions is satisfied:
$$\binom{n}{2}\left(\poneone - \delta_3\right)\leq \max_{\ggb' \cong \stbrel} \sum_{i<j}\ea_{i,j}\eb'_{i,j} \leq \binom{n}{2}\left(\poneone + \delta_3\right),$$
if $\poneone\pzerozero > \ponezero\pzeroone$, or 
$$\binom{n}{2}\left(\poneone - \delta_3\right)\leq \min_{\ggb' \cong \stbrel} \sum_{i<j}\ea_{i,j}\eb'_{i,j} \leq \binom{n}{2}\left(\poneone + \delta_3\right),$$
if $\poneone\pzerozero < \ponezero\pzeroone$, where $\delta_3:=\delta/ \lvert \log \frac{\poneone\pzerozero}{\ponezero\pzeroone}\rvert$.

%% file: Discussion.tex
In this work, we established the minimum achievable compression rate when the object to compress and the object available as side information are graphs or structures (unlabelled graphs). Unlike a labelled graph, a structure does not lend itself to a simple analysis. However, we can rely on the observation that structures retain most of the information about the labelled graphs. For an Erd\H{o}s-R\'enyi graph, the entropy of a structure is roughly $n\log n$ bits less than the entropy of a labelled graph, which is of the order $n^2$ when the model parameters do not depend on the number of vertices $n$. 

One of the key arguments in this paper is that the joint probability  involving a structure, which is the sum of the joint probabilities over all labellings, can be approximated by the maximum of the joint probabilities over all labellings. We have shown that this maximum in turn is close in some sense to the true underlying labelling of the structure. As a result, we were able to show that the asymptotics of the conditional distributions involving structures are the same as those of the distributions involving only labelled graphs, yielding $R^*= H(\ea|\eb)$.

Our results readily extend to distributed compression involving unlabelled graphs. In this situation, Theorem~\ref{thm:asymp:prob:gen} and Theorem~\ref{thm:asymp:inprob:weighted} have to be extended for the problem under consideration. These results can then be used to construct an encoding and decoding scheme. In the current work, we considered the dense regime of the CER model, where $P_{\ea, \eb}$ does not depend on $n$. It is also of interest to consider other regimes where $P_{\ea, \eb}$  does depend on $n$. We believe that the current techniques should be extendable to other regimes as well, but we leave these questions for future study.

The CER model is rather ideal when considered in the context of practical graphs. Nevertheless, the model captures first-order correlations between graphs and can guide the choice of schemes to implement for real-world graphs. In addition to the CER model, it is also of interest to study graph compression with side information for random graph models such as the stochastic block model, where the connectivity of vertices is determined by the latent community membership variables. Furthermore, we can also study a more realistic scenario where the correlated graphs do not share the same vertex set but do possess some overlap of vertices.